\documentclass[10pt,conference]{IEEEtran}
\usepackage{amsfonts}
\usepackage{amsmath}
\usepackage{amssymb}
\usepackage{amsthm}
\usepackage{lscape}
\usepackage{epsf}
\usepackage{graphicx}
\usepackage{verbatim}
\usepackage{color} 
\usepackage{bbold}
\usepackage{tikz}
\usetikzlibrary{arrows}
 \usepackage{nopageno}

\newtheorem{theorem}{\indent Theorem}[section]
\newtheorem{lemma}[theorem]{\indent Lemma}
\newtheorem{corollary}[theorem]{\indent Corollary}

\newtheorem{proposition}[theorem]{\indent Proposition}

\newtheorem{question}[theorem]{\indent Question}
\newtheorem{EXAMPLE}{\indent Example}[section]
\newtheorem{definition}{\indent Definition}[section]

\newenvironment{example}{\begin{EXAMPLE}\rm}{\rm\end{EXAMPLE}}

\newcommand{\es}{{\varepsilon}}

\newcommand{\nn}{{\mathbb N}}

\newlength{\Algwidth}

\title{Permutation codes, source coding and a generalisation of Bollob\'as-Lubell-Yamamoto- Meshalkin and Kraft inequalities}
\author{{\bf Kristo Visk and Ago-Erik Riet} \\ 
Institute of Mathematics and Statistics \\
Faculty of Science and Technology,
University of Tartu, J. Liivi 2, Tartu 50409, Estonia}
\begin{document}


\maketitle

\begin{abstract}
We develop a general framework to prove Kraft-type inequalities for prefix-free permutation codes for source coding with various notions of permutation code and prefix. We also show that the McMillan-type converse theorem in most of these cases does not hold, and give a general form of a counterexample. Our approach is more general and works for other structures besides permutation codes. The classical Kraft inequality for prefix-free codes as well as results about permutation codes follow as corollaries of our main theorem and main counterexample. 
\end{abstract}


\section{Introduction}
Non-volatile memory is a type of computer memory that can store information after the device has been turned off. Flash memory is a type of non-volatile storage device. Data is stored onto a flash memory device by injecting charges into memory cells. It is possible to increase the level of charge of a particular cell, but to decrease the level of charge it is required to erase and overwrite a large block of cells.

Over time the drift of  electrical charges in memory cells might occur. Drift may occur at different rates for different cells, which makes sustaining required charge levels difficult, because charge levels in every cell would have to be monitored separately. Moreover, while increasing the charges, some cells might get overcharged, resulting in an overshoot error. Since reducing the charge levels is a complex process, the reliability of flash memory devices decreases over time.

To manage charge levels in memory cells more efficiently, multi-level memory cells are used. Single-level memory cells at any given time are either charged or empty, while in a multi-level cell system the charge of an individual cell can have more than two different levels. If the charge levels of cells in a block of cells are different, i.e.~arranged as a permutation, the drift and overshoot errors become easier to detect. By coding each block of cells into permutations and managing injections of charges it is possible to reduce drift and overshoot errors, using the rank modulation scheme proposed in~\cite{rank modulation}. Permutation codes for error correction for flash memories which provide further robustness have been studied for example in~\cite{Blake19791, gad, 5485013, bmz, farnoud2013, gologlu2015}. Permutation codes were also proposed for use in powerline communications, see for example~\cite{slepian}.

In this work, we study a question of unique decoding of permutation codes. This question is important for understanding theoretical limits for optimal (source) coding of information stored in flash memories. 

To this end, we prove a generalisation of the Bollob\'as-Lubell-Yamamoto-Meshalkin (LYM)~\cite{bol1965, lubell, meshalkin, yamamoto} and Kraft~\cite{kraft} inequalities, which works for certain graded posets. As corollaries of our main theorem we obtain that a Kraft-type inequality holds for prefix-free permutation codes in different contexts, where we give several definitions of permutation codes and several definitions of what it means to be a `prefix'. As corollaries to our general counterexample, we obtain that a McMillan-type converse theorem fails in most of these cases, but not for the case of the classical notion of prefix. 

\section{Notation}
Denote $[k]:=\{1, 2, \dots ,k\}$, let $\mathbb{N}:=\{1,2, \dots \}$ and $\mathbb{N}_0:=\{0\}\cup \mathbb{N}$. For a set $S$, its cardinality is denoted $\# S$, sometimes also written $|S|$.

An \emph{alphabet} is a finite set $S:=\{s_{i}\mid i\in [n]\}$ of cardinality $n\neq 0$. A \emph{symbol} is an element $s \in S$, a finite sequence of symbols $s_{1} \dots s_{k}$  is a \emph{string}. The \emph{length} of a string is the number of symbols it consists of. Let  
$
S^{l}:=\{s_{1} \dots s_{l}\mid s_{j}\in S \mbox{ for all } j\in [l]\} 
$
denote the set of strings of length $l$. Define $S^0 := \{\es\}$ where $\es$ is the unique string of length 0 called the \emph{empty string}. Let
$
S^{*}:=\bigcup_{j\in\mathbb{N}_0}{S^{j}}= S^{0} \cup S^{1} \cup  \dots 
$
be the set of all finite strings. For strings $t,u\in S^*$ where $t=s_{i_1}s_{i_2}\ldots s_{i_k} \in S^k$ and $u=s_{j_1}s_{j_2}\ldots s_{j_r} \in S^r$ their \emph{concatenation} is $tu=s_{i_1}\ldots s_{i_k} s_{j_1} \ldots s_{j_r} \in S^{k+r}$. A string $t$ is a \emph{prefix} to the string $u$ if $u=tw$ for some string $w$. If $w \not= \epsilon$, then $t$ is called a \emph{proper prefix} of $u$. 

 A \emph{permutation} of the set $[k]$ is a bijection $\sigma : [k] \rightarrow [k]$. To denote permutations $\sigma: i\mapsto b_{i}$ $ (i \in [k])$ we use vectors $(b_{1}, \dots ,b_{k})$ as well as strings $b_{1} \dots b_{k}$. We also denote $\mathbb{S}_{k}:=\{\sigma\mid\sigma: [k] \rightarrow [k] \mbox{ is a bijection} \}$ for the set (in fact, group) of permutations on $[k]$. For example $\sigma=2314=(2,3,1,4)\in \mathbb{S}_4$ is the permutation $\sigma(1)=2,\; \sigma(2)=3,\; \sigma(3)=1,\;\sigma(4)=4$. 

We use notation
$$
T_{k}^l:=\{(n_{1}, \dots ,n_{l})\mid \forall j: n_{j}\in [k] ; \; i \neq j \Rightarrow n_{i}\neq n_{j}\}
$$
for the set of $l$-element \emph{partial permutations} on the set $[k]$, i.e.~the injective mappings $\tau : [l] \rightarrow [k]$. Let $$T_{k}:=T_k^1 \cup T_k^2 \cup \ldots \cup T_k^k$$
and write
$$
\mathbb{T}_{k}:=\{(n_{1}, \dots ,n_{l})\mid \forall j: n_{j}\in [l], l\leq k ; \; i \neq j \Rightarrow n_{i}\neq n_{j}\}
$$
$$=\mathbb{S}_1\cup \mathbb{S}_2 \cup \ldots\cup\mathbb{S}_k.$$ We say $\sigma\in \mathbb{S}_l$ is \emph{the pattern of} $\tau\in T_k^l$ if the relative ordering of symbols is the same, i.e.~if, for all $i,j\in [l],$ we have $\sigma(i) < \sigma(j)$ if and only if $\tau(i) < \tau(j)$. A permutation $\sigma \in \mathbb{S}_l$ is \emph{a pattern in} a partial permutation $\tau\in T_k^m$ if there are indices $1\leq b_1 < b_2 < \ldots < b_l \leq m$ such that, for all $i,j\in [l],\;\;\sigma(i) < \sigma(j)$ if and only if $\tau(b_i) < \tau(b_j)$. We call $\sigma\in T_k^m$ a (not necessarily consecutive) \emph{subsequence} of $\tau\in T_k^l$ if $m\leq l$ and there are $m$ indices $1\leq b_1 < b_2 < \ldots < b_m \leq l$ with $\sigma(i) = \tau(b_i)$ for all $i \in [m]$. Note that $\sigma$ is a pattern in $\tau\in T_k$ if and only if it is the pattern of a subsequence of $\tau$. 

For example $253\in T_6^3$ is a subsequence in $\mathbf{25}1\mathbf{3} \in T_6^3$, the pattern of $253\in T_6^3$ is $132\in \mathbb{S}_3$ and thus $132\in \mathbb{S}_3$ is a pattern in $2513 \in T_6^3$. 

We call $\sigma\in T_k^m$ a (consecutive) \emph{substring} of $\tau\in T_k^l$ if $m\leq l$ and there is $0\leq n \leq l-m$ with $\sigma(i) = \tau(i+n)$ for all $i \in [m]$. If $\sigma$ is a substring of $\tau$ then it is also a subsequence. 

For example $51\in T_6^2$ is a substring in $2\mathbf{51}3 \in T_6^4$ but the subsequence $253\in T_6^3$ is a not a substring in $\mathbf{25}1\mathbf{3} \in T_6^4$. 

We say $\sigma \in \mathbb{S}_m$ is a \emph{substring pattern} of $\tau \in T_k$ if it is the pattern of a substring of $\tau$. So $21\in \mathbb{S}_2$ is a substring pattern in $2\mathbf{51}3 \in T_6^4$ since it is the pattern of the substring $51\in T_6^2$.

\section{Codes, their prefix-freeness and unique decodability}
Let  $S_{1}$ and $S_{2}$ be alphabets. A (non-singular, classical) \emph{code} is an injective map $c : S_{1} \rightarrow S_{2}^{*}$. An image $c(s)$ is called a \emph{codeword} corresponding to symbol $s \in S_1$.
The definition of a code extends to all strings as $c(s_{1} \dots s_{n}):=c(s_{1}) \dots c(s_{n})$. It is easily checked the new map $c : S_1^* \rightarrow S_2^*$, called the \emph{extension} of the code $c$, is well-defined. A code $c : S_{1} \rightarrow S_{2}^{*}$ is \emph{uniquely decodable} if its extension is injective, i.e.~if $c(s_{1} \dots s_{n}) = c(s'_{1}\dots s'_{m})$ implies $m=n$ and $s_{j} = s'_{j}$ for all $j \in [n]$. A code $c : S_{1} \rightarrow S_{2}^{*}$ is \emph{prefix-free} if there do not exist $s_i,s_j\in S_1$ with $c(s_i)$ a proper prefix of $c(s_j)$. 

It is easy to see that a prefix-free code is uniquely decodable: we can read symbols in an output string $c(s_{1} \dots s_{n})=c(s_{1}) \dots c(s_{n})$ from left to right, and prefix-freeness guarantees that no proper prefix of $c(s_1)$ is a codeword and also that $c(s_1)$ is not a proper prefix to any codeword. Hence, encountering the substring $c(s_1)$ at the beginning of the output $c(s_{1} \dots s_{n})$, we are guaranteed that it arose by encoding the symbol $s_1$. We then continue by decoding $c(s_{2} \dots s_{n})$ similarly. Because of this property, prefix-free codes are also called \emph{instantaneous}, see for example~\cite{cover} (Ch.~5).

Let $c : S_1 \rightarrow S_2^*$ be a code. Then the sequence $(a_0,a_1,\ldots)$, where $a_j := \#\{s\in S_1 \mid c(s) \in S_2^j\}, \; j\in \nn_0$, is the \emph{parameter sequence} of the code $c$. The parameters count the codewords of each length.

We shall now state the Kraft inequality~\cite{kraft, mcmillan}, see also~\cite{cover} (Ch.~5) from classical source coding, which holds for all uniquely decodable classical codes, in particular for all prefix-free classical codes.

\begin{proposition}\label{kraft}
Let $c : S_1 \rightarrow S_2^*$ be a uniquely decodable classical code with parameter sequence $(a_0,a_1,\ldots)$ and $\# S_2 = r$. Then
\begin{equation} \label{kraft eqn} K_c := \sum_{i=0}^\infty \frac{a_i}{{r^i}} \leq 1.\end{equation}
\end{proposition}

The number $K_c$ is the \emph{Kraft number}, also known as the Kraft sum or the Kraft-McMillan number. Note that Proposition~\ref{kraft} states that the sum of densities of used codewords of a fixed length, over all lengths, is at most 1. 

McMillan~\cite{mcmillan} proved the following strong converse of Proposition~\ref{kraft}, known as (the converse part of) McMillan Theorem. We remark that its proof is the construction of a code by picking vertices in the $r$-ary code tree greedily, starting with vertices closer to the root and going in the lexicographic order.

\begin{proposition}\label{mcmillan}
Let $r\in \mathbb{N}$. If non-negative integers $a_i\in \mathbb{N}_0 \cup \{0\},\; i\in \mathbb{N}_0$, are such that inequality~(\ref{kraft eqn}) holds, then there exists a prefix-free code $c : S_1 \rightarrow S_2^*$ with $\# S_2 = r$ and parameter sequence $(a_0,a_1,\ldots)$. 
\end{proposition}

\section{Permutation codes}
We shall define permutation codes in two ways, by restricting the output space of classical codes.

\begin{definition} Let $S$ be an alphabet and $k\in \mathbb{N}$. We define a \emph{permutation code} as an injection $c : S \rightarrow T_k$. Note that 
$$ \# T_k = \sum_{l=1}^k {k\choose l} l!.$$
We have $\# S \leq \sum_{l=1}^k {k\choose l} l!$ because of the injectivity of $c$. The \emph{parameter sequence} of the code is $(a_0,a_1,\ldots )$, where $a_j := \# \{ s\in S \mid c(s) \in T_k^j\}$, as for classical codes. 
\end{definition}

\begin{definition}
A more restrictive definition of a \emph{permutation code} is an injection $c : S \rightarrow \mathbb{T}_k$. Note that 
$$ \# \mathbb{T}_k = \sum_{l=1}^k l!.$$ We have $\# S \leq \sum_{l=1}^k l!$ because of the injectivity of $c$. The \emph{parameter sequence} of the code is then $(a_0,a_1,\ldots)$ where $a_j := \# \{ s\in S \mid c(s) \in \mathbb{S}_j\}$.
\end{definition}

\subsection{Definitions of `prefix-freeness' for permutation codes}
Let us define the \emph{permutation constant} of a permutation code $c$ in these cases respectively as $$ P_c := \sum_{l=1}^k \frac{a_l}{{k\choose l}l!},\;\;\; \text{or,}\;\;\; \mathbb{P}_c := \sum_{l=1}^k \frac{a_l}{l!} .$$ This is the analogue of the Kraft number from classical codes. 

The extension, prefix-freeness and unique decodability of a permutation code $c : S \rightarrow T_k$ or $c : S \rightarrow \mathbb{T}_k$ is understood as that notion for the same code viewed as a code $c : S \rightarrow [k]^*$. 

Now we give some notions analogous to prefix-freeness. A permutation code $c : S \rightarrow T_k$ or $c : S \rightarrow \mathbb{T}_k$ is \emph{\{subsequence-, substring-, pattern- or substring-pattern-\}free} if there are no two different codewords $c(s_1) \not= s(s_2)$ such that $c(s_1)$ is respectively a \{subsequence, substring, pattern or substring pattern\} in $c(s_2)$. These notions can also be defined for classical codes $c : S_1 \rightarrow S_2^*,$ with the subsequence- and substring-freeness being perhaps the more natural notions.

We shall see that often for prefix-free, subsequence-free, substring-free, pattern-free or substring-pattern-free permutation codes, $P_c \leq 1$, or, $\mathbb{P}_c \leq 1$, i.e.~the analogue of Proposition~\ref{kraft} holds. That is, the sum of densities of used codewords of fixed length, over all lengths, is at most 1.

However, we shall also see that $P_c \leq 1$, or, $\mathbb{P}_c \leq 1$ for given code parameters does not in general imply that a subsequence-free, substring-free, pattern-free or substring-pattern-free permutation code with these parameters exists. In some of these cases there is no analogue of Proposition~\ref{mcmillan}.

\section{A generalisation of the LYM and Kraft inequalities}

We shall state the Bollob\'as-Lubell-Yamamoto-Meshalkin inequality, also known as the LYM inequality~\cite{bol1965, lubell, meshalkin, yamamoto}, see also~\cite{jukna} (Ch.~8). Consider the set of all subsets of a finite set $[n]$, denoted by $\mathcal{P}([n])$, ordered by the subset relation $\subseteq$. This is a poset. Denote $[n]^{(k)} := \{A \subseteq [n] : \# A = k\}$. 

\begin{proposition}\label{lym} Suppose $\mathcal{A} \subseteq \mathcal{P}([n])$ is an antichain, i.e.~$A,B\in \mathcal{A}$ and $A \subseteq B$ implies $A = B$. Then $$\sum_{i=0}^n \frac{\# (\mathcal{A} \cap [n]^{(i)})}{{n\choose i}} \leq 1.$$
\end{proposition}

That, is, the sum of densities of $\mathcal{A}$ in each level, summed over all levels, is at most 1.

Now we shall prove a common generalisation of the Proposition~\ref{kraft} for prefix-free codes and Proposition~\ref{lym}. As consequences we obtain some analogues of Proposition~\ref{kraft}, namely $P_c \leq 1,$ and, $\mathbb{P}_c \leq 1$, for some of prefix-free, subsequence-free, substring-free, pattern-free and substring-pattern-free permutation codes. We remark that our theorem follows from the so-called AZ identity for general finite posets~\cite{ahlswede} but our proof here is self-contained and more suited for the applications we have in mind. We also remark that our proof follows closely the known proof of the Proposition~\ref{lym} via the Local LYM inequality, and our framework of level-regular graded posets, to be defined, is chosen so that this proof still works, while being general enough for the corollaries we have in mind.

We shall also investigate when a converse statement such as Proposition~\ref{mcmillan} can hold. For example for Proposition~\ref{lym} such a converse statement does not hold in general. 

\subsection{Level-regular graded posets}

We shall consider a special kind of graded posets. A \emph{partially ordered set} or a \emph{poset} is a set $P$ together with a binary relation $ \leq $ that is reflexive, transitive and antisymmetric, i.e.~for all $a\in P$, $a\leq a$, for all $a,b,c\in P$, if $a\leq b$ and $b\leq c$ then $a\leq c$ and for all $a,b\in P$, if $a \leq b$ and $b\leq a$ then $a=b$. Write $a<b$ for $a\leq b$ and $a\not= b$. We say \emph{$b$ covers $a$} if $a < b$ and there is no $c\in P$ with $a<c$ and $c<b$. An element $a\in P$ is \emph{minimal} if there is no $b\in P$ with $b < a$. A \emph{graded poset} is a poset $P$ with a \emph{rank function} $\rho : P \rightarrow \mathbb{N}_0$ satisfying $\rho(c) = 0$ for all minimal $c\in P$, and, $\rho(b) = \rho(a) + 1$ if $b$ covers $a$, and, if $a < b$ then $\rho(a) < \rho(b)$. 

A \emph{directed (multi)graph} $G=(V,E)$ is a set $V,$ called its \emph{vertex set}, together with a multiset $E \subseteq V\times V$ of ordered pairs of vertices, called its \emph{edge set} (there may be multiple edges $(u,v)$ for fixed $u,v\in V$). An element $v\in V$ is a \emph{vertex} and an element $e\in E$ is an \emph{edge}. An edge $e=(u,v)$ is \emph{directed from $u$ to $v$} or goes from $u$ to $v$. A vertex $v$ is a \emph{neighbour} of a vertex $u$ if there is an edge $(u,v)$ or $(v,u)$. A directed graph is \emph{weakly connected}, if its underlying undirected graph is connected, i.e.,~without regard to directions of edges, one can walk from any vertex to any other vertex along edges (we can walk from a vertex to any of its neighbours). The \emph{up-degree} of a vertex $u$ is $\# \{e\in E \mid \exists v\in V:\; e=(u,v)\},$ i.e.~the number of edges directed from $u$, and the \emph{down-degree} of a vertex $v$ is $\# \{e\in E \mid \exists u\in V:\; e=(u,v)\},$ i.e.~the number of edges directed to $v$. Sometimes the up- or down-degree is just called degree. For a graph $G=(V,E)$ and a subset $V'\subseteq V$, the graph $G'=(V',E\cap (V'\times V'))$ is called an \emph{induced subgraph} of $G$, i.e.~we keep all edges with both endpoints in $V'$ and only them; then $G'$ is \emph{induced by} $V'$.

The \emph{Hasse diagram} of a graded poset $P$ is a directed graph with vertex set $P$, and an edge from $a$ to $b$ if and only if $b$ covers $a$; it is drawn with elements of the same rank on the same horizontal level and elements of higher ranks higher. 

\begin{example}
See Figure~1 for the Hasse diagram of the poset of subsets of $\{1,2\}$ with the subset relation $\subseteq$.

\begin{figure}[h]
\begin{center}
\begin{tikzpicture}
  \tikzset{edge/.style = {->,> = latex'}}
  \node (max) at (0,1) {$\{1,2\}$};
  \node (d) at (-1,0) {$\{1\}$};
  \node (f) at (1,0) {$\{2\}$};
  \node (min) at (0,-1) {$\emptyset$};
\draw[edge,line width=1pt] (min) to (d);
\draw[edge,line width=1pt] (min) to (f);
\draw[edge,line width=1pt] (d) to (max);
\draw[edge,line width=1pt] (f) to (max);
\end{tikzpicture}
\end{center}
\caption[Subsets]{The Hasse diagram of the poset of subsets of $\{1,2\}$.}
\end{figure}
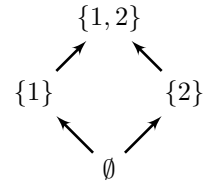

\end{example}

Let us say a graded poset $P$ is \emph{level-regular} if the bipartite (multi)graph induced by any two consecutive levels of its Hasse diagram is biregular --- that is, all elements on the same level, i.e.~with the same rank, are covered by an equal number of elements, and also cover an equal (perhaps different) number of elements (with multiplicity).

Let $\mathcal{A} \subseteq P$ be any set of elements of the same rank, i.e.~$\rho(a) = \rho(b)$ for all $a,b\in\mathcal{A}$. Then its \emph{upper shadow} $\delta^+(\mathcal{A}) := \{b \mid \exists a \in \mathcal{A} : b\; \text{covers} \; a\}$ is the set of all elements covering some element of $\mathcal{A}$, and its \emph{lower shadow} $\delta^-(\mathcal{A}) := \{b \mid \exists a\in\mathcal{A} : a\; \text{covers} \; b\}$ is defined analogously.

An \emph{antichain} is a subset $A\subseteq P$ whose elements are pairwise incomparable, i.e.~$a\not< b$ for all $a,b\in A$.

\subsection{A common generalisation of the LYM and Kraft inequalities}

Our main Theorem is the following generalisation of Proposition~\ref{lym} and Proposition~\ref{kraft}. 

\begin{theorem}\label{graded} Let $P$ be a level-regular graded poset and let $A \subseteq P$ be an antichain. Denote $P^{(i)} := \{p\in P \mid \rho(p)=i\}$ and $A^{(i)} := \{a\in A \mid \rho(a)=i\}$. Assume all levels are finite, i.e.~$\# P^{(i)} < \infty$ for each $i\in\mathbb{N}_0$. Define the \emph{LYM number} of the antichain as $$ L_A := \sum_{i=0}^\infty \frac{\# A^{(i)}}{\# P^{(i)} }.$$ Then 
$$L_A \leq 1.$$
\end{theorem}

\begin{proof}
First let us assume that the poset $P$ is finite.
Then the LYM number is a finite sum. We shall proceed by induction on $k := \max\{\rho(a) \mid a\in A\}$. If $k=0$ then $\# A^{(i)} = 0$ for all $i>0$ and $$\sum_{i=0}^\infty \frac{\# A^{(i)}}{\# P^{(i)} } = \frac{\# A^{(0)}}{\# P^{(0)}} \leq 1,$$ as needed. Let us define $A' := \{a \in A \mid \rho(a) < k\} \cup \{p \in P \mid \rho(p)=k-1, \; p\in \delta^-(A)\}$. We shall prove that, on replacing $A$ by $A'$, its LYM number does not decrease, i.e.~$L_{A} \leq L_{A'}$, and that $A'$ is still an antichain. The claim now follows by induction, since $L_{A'} \leq 1$ by the induction hypothesis.

Let $a,b\in A'$. We shall prove that $a\not< b$. If $a,b\in A$ then $a\not< b$ as $A$ is an antichain. If $a,b\in P$ with $\rho(a) = \rho(b)$ then also $a\not< b$ by the definition of a graded poset. Suppose for a contradiction that $a < b$. The only way it might happen is with $a\in A$ and $b\in \delta^-(A)$ with $\rho(b) = k-1$. But since $b$ is in the lower shadow of $\{a \in A \mid \rho(a) = k\}$, there exists $c\in A$ with $b < c$. By transitivity, $a < c$ with $a,c\in A$ --- a contradiction with $A$ being an antichain. Hence $A'$ is an antichain.

Note that $\delta^-(A^{(k)}) \cap A^{(k-1)} = \emptyset$, as $A$ is an antichain. Hence, to prove that $L_{A'} \geq L_A$, it is enough to show that 
\begin{lemma}\label{local lym}
$$\frac{\# \delta^-(A^{(k)}) }{\# P^{(k-1)}} \geq \frac{\# A^{(k)}}{\# P^{(k)}}.$$
\end{lemma}
This is the analogue of what is known as the Local LYM inequality, which reads that shadows have greater density.
\begin{proof}
We shall prove the Lemma by degree considerations of the Hasse diagram of levels $k$ and $k-1$. 
Let the down-degree, i.e.~the number of elements it covers, of each $v\in P^{(k)}$ be $d$, and the up-degree, i.e.~the number of elements covering it, of each $w \in P^{(k-1)}$ be $u$. The number of edges between the sets $P^{(k)}$ and $P^{(k-1)}$ is $d \cdot \# P^{(k)} = u \cdot \# P^{(k-1)}$. The number of edges in the subgraph induced by $A^{(k)}$ and $\delta^-(A^{(k)})$ is equal to $d\cdot \# A^{(k)}$ on the one hand and at most $u \cdot \# \delta^-(A^{(k)})$ on the other hand --- a vertex of the original has $d$ neighbours in the lower shadow and a vertex of the lower shadow has at most $u$ neighbours in the original. Hence $$\frac{\# P^{(k)}}{\# P^{(k-1)}} = \frac{u}{d} \geq \frac{\# A^{(k)}}{\# \delta^-(A^{(k)})},$$
and rearranging proves the Lemma.
\end{proof}
This proves the Theorem if $P$ is finite. To prove the infinite case, restrict the poset to levels up to $N$, for every $N\in\mathbb{N}$, and then by the finite case we have $$L_A = \sum_{i=1}^\infty \frac{\# A^{(i)}}{\# P^{(i)} } = \lim_{N\rightarrow\infty}\sum_{i=1}^N \frac{\# A^{(i)}}{\# P^{(i)} } \leq 1.$$
\end{proof}

\section{Failure of the converse McMillan theorem}

We may ask about the analogue of Proposition~\ref{mcmillan} in this general setting. 

\begin{question}\label{false}
Let $P$ be an (infinite or finite) level-regular graded poset. Assume that all levels are finite, i.e.~$\# P^{(i)} < \infty$ for each $i\in\mathbb{N}_0$. Let $a_i\in\mathbb{N}_0$ for each $i\in\mathbb{N}_0$. Assume $$\sum_{i=1}^\infty \frac{a_i}{\# P^{(i)} } \leq 1.$$ Is it true that then there exists an antichain $A\subseteq P$ with $\# A^{(i)} = a_i$ for each $i$?
\end{question}

The answer is in general ``No'' but ``Yes'' for example if the Hasse diagram is a tree. The next Theorem gives the general form of our counterexample to Question~\ref{false}. We can hope for a counterexample only when Lemma~\ref{local lym} is not tight, i.e.~when
${\# \delta^-(A^{(k)}) }/{\# P^{(k-1)}} > {\# A^{(k)}}/{\# P^{(k)}}.$

\begin{theorem}\label{counter} Consider two consecutive levels $P^{(i)}$ and $P^{(i+1)}$ of a level-regular graded poset $P$. Suppose that the equal up-degrees from level $i$ to level $i+1$ are $u > 1$ and the equal down-degrees from level $i+1$ to level $i$ are $d > 1$, i.e.~each element of $P^{(i)}$ is covered by $u$ elements of $P^{(i+1)}$ and that each element of $P^{(i+1)}$ covers $d$ elements of $P^{(i)}$. Assume that the graph, induced by levels $i$ and $i+1$ of the Hasse diagram, is weakly connected. Further assume that the greatest common divisor $\gcd(\# P^{(i)}, \# P^{(i+1)}) =: g > 1$.

Define $a_i :=  \frac{g-1}{g}\# P^{(i)};\; a_{i+1} :=  \frac1g \# P^{(i+1)};\; a_j := 0$ for $j\not= i,i+1,$ and note that they are integers. Then $$\sum_{i=1}^\infty \frac{a_i}{\# P^{(i)} } \leq 1$$ but there is no antichain $A$ with $\# A^{(j)} = a_j$ for all $j\in \mathbb{N}_0$. 
\end{theorem}

\begin{proof}
Fix any subset $A^{(i+1)} \subseteq P^{(i+1)}$ with $\# A^{(i+1)} = a_{i+1} = \frac{1}{g}\# P^{(i+1)}$. Consider the graph $G_i=(V_i,E_i)$ induced by $V_i := A^{(i+1)}\cup\delta^-(A^{(i+1)})$. It is a bipartite graph with parts $A^{(i+1)}$ and $\delta^-(A^{(i+1)})$, i.e.~all edges $e\in E_i$ are of the form $(v,w)$ with $v\in \delta^-(A^{(i+1)})$ and $w\in A^{(i+1)}$. The average of down-degrees of vertices in $A^{(i+1)}$ is equal to $d$ since all neighbours of vertices of $A^{(i+1)}$ lying in the $i$-level of the Hasse diagram of $P$ are contained in the graph $G_i$. The average of up-degrees of vertices in $\delta^-(A^{(i+1)})$, however, is strictly less that $u$. To see this, note that all degrees are at most $u$. But there is a vertex of degree less than $u$ since not all neighbours of all vertices of $\delta^-(A^{(i+1)})$ in the graph induced by $P^{(i)}\cup P^{(i+1)}$ (where all these degrees are $u$) lie in the set $A^{(i+1)}$. Indeed, otherwise the graph induced by $P^{(i)}\cup P^{(i+1)}$ would not be weakly connected, with $\delta^-(A^{(i+1)})\cup A^{(i+1)}$ being disconnected from the rest of the graph. Hence, $$\frac{\# \delta^-(A^{(i+1)}) }{\# P^{(i)}} > \frac{\# A^{(i+1)}}{\# P^{(i+1)}}.$$ Thus $a_i > \# (P^{(i)} \backslash \delta^-(A^{(i+1)}))$ and we cannot pick $a_i$ elements to our antichain $A$ on level $P^{(i)}$. Since the choice of $a_{i+1}$ elements comprising $A^{(i+1)}$ was arbitrary, there is no antichain $A$ with the required properties.
\end{proof}

\begin{corollary}
The Kraft inequality~(\ref{kraft eqn}), $K_c \leq 1,$ holds for a classical code $c : S_1 \rightarrow S_2^*,\; \# S_2 = r,$ if it is prefix-free, subsequence-free or substring-free. For $r\geq 2$ there exist parameter sequences $(a_0,a_1,\ldots)$ for which $K_c \leq 1$ but there is no subsequence-free or substring-free code $c : S_1 \rightarrow S_2^*$.
\end{corollary}

\begin{proof}
Consider the poset of all possible codewords, with the prefix, subsequence or substring relation. It is a graded poset: the rank of a codeword is its length. It is level-regular: the number of codewords $c\in S_2^{l+1}$ covering a codeword of length $l$ is $r = \# S_2$ for the prefix, $(l+1)r$ for the subsequence and $2r$ for the substring relation, and a codeword of length $l$ covers 1 codeword of $S_2^{l-1}$ for the prefix, $l$ codewords for the subsequence and $2$ codewords for the substring relation (with multiplicity: the addition or removal of a different symbol or in a different position may produce equal outputs). Hence Theorem~\ref{graded} proves Kraft inequality~(\ref{kraft eqn}) for these codes.

Conversely, for the subsequence and substring relations, note that level 1 of the poset has $r$ elements, level 2 has $r^2$ elements, with $\gcd(r,r^2) = r > 1$. The graph induced by levels 1 and 2 of the Hasse diagram is weakly connected: any string can be transformed into any other string of the same length by alternate adding and removing of symbols at the beginning or end. Hence Theorem~\ref{counter} shows the analogue of Proposition~\ref{mcmillan} fails in these cases.
\end{proof}

\begin{corollary}
We have $P_c \leq 1$ for \{prefix-, subsequence- or substring-\}free permutation codes $c : S \rightarrow T_k$  and $\mathbb{P}_c \leq 1$ for \{prefix-, subsequence-, substring-, pattern- or substring-pattern-\}free permutation codes $c : S \rightarrow \mathbb{T}_k$. For all $k\geq 3,$ there are parameter sequences $(a_0,a_1,\ldots)$ with $P_c \leq 1$ but no subsequence-free or substring-free permutation codes $c : S \rightarrow {T}_k$, and, parameter sequences $(a_0,a_1,\ldots)$ with $\mathbb{P}_c \leq 1$ but no pattern-free or substring-pattern-free codes $c : S \rightarrow \mathbb{T}_k$.
\end{corollary}

\begin{proof}
Sketch. For the pattern- and substring-pattern relations we need a trick to obtain level-regularity: define a new poset, adding a new level between every pair of consecutive levels, i.e.~look at the pattern relation in two steps: first pick a subsequence or substring in $T_k$ and then consider its pattern.
\end{proof}

The subsequence relation has relevance to Ulam codes~\cite{ulam, farnoud2013, gologlu2015}. A permutation code $c : S \rightarrow T_k$ has minimum Ulam distance $d$ if and only if every string $s\in [k]^{k-d+1}$ is a subsequence in at most one codeword. For fixed-length codes this was explored for example in~\cite{Levenshtein, Mathon, farnoud2013, gologlu2015}.

\section{Open problems}

Does the analogue of Kraft inequality~(\ref{kraft eqn}) hold with definitions of `unique decodability', expanding the freeness definitions, for classical or permutation codes? Does there exist an analogue of Huffman coding for these types of codes?

\section{Acknowledgements}

This work was supported in part by the Estonian Research Council through the research grants PUT405, PUT620 and IUT20-57. The authors wish to thank Vitaly Skachek for helpful discussions.


\begin{thebibliography} {99}




\bibitem{ahlswede}
R.~Ahlswede. and Z.~Zhang, ``An identity in combinatorial extremal theory,'' \emph{Advances in Mathematics}, vol.~80, no.~2, pp. 137--151, 1990.


\bibitem{5485013}
A.~Barg and A.~Mazumdar, ``Codes in permutations and error correction for rank
  modulation,'' \emph{IEEE Trans. on Inform. Theory}, vol.~56, no.~7, pp. 3158
  --3165, Jul. 2010.

\bibitem{Blake19791}
I.F. Blake, G.~Cohen, and M.~Deza, ``Coding with permutations,''
  \emph{Information and Control}, vol.~43, no.~1, pp. 1--19, 1979.

  
  \bibitem{bol1965}
B.~Bollob\'as, ``On generalized graphs,''  \emph{Acta Math. Acad. Sci. Hung.,} vol.~16, no.~3--4, pp. 447--452, 1965.
	
\bibitem{cover}
T.~Cover, and J.A.~Thomas, ``Elements of information theory,'' \emph{Wiley-Interscience New York, NY, USA,} 1991.

\bibitem{farnoud2013}
F.~Farnoud, V.~Skachek, and O.~Milenkovic, ``Error-correction in flash memories via codes in the Ulam metric,'' 
\emph{IEEE Trans. on Inform. Theory}, vol. 59, no. 5, pp. 3003--3020, May 2013.

\bibitem{gad}
  E.~En~Gad, M.~Langberg, M.~Schwartz, and J.~Bruck, ``Constant-weight Gray codes for local rank modulation,'' \emph{IEEE Trans. on Inform. Theory,} vol.~57, no.~11, pp. 7431--7442, Nov. 2011.

\bibitem{gologlu2015}
F.~G\"olo\u{g}lu, J.~Lember, A.-E.~Riet, and V.~Skachek, ``New bounds for permutation codes in Ulam metric,'' \emph{Proc. IEEE Intern. Symp. on Inform. Theory}, pp. 1726--1730, Jun. 2015.


\bibitem{rank modulation}
A.~Jiang, R.~Mateescu, M.~Schwartz, and J.~Bruck,
   ``Rank modulation for flash memories,'' in \emph{IEEE Trans. on Inform. Theory}, vol.~55, no.~6, pp. 2659 --2673, June 2009. 

 
\bibitem{jukna}
S.~Jukna, ``Extremal Combinatorics: With Applications in Computer Science,'' \emph{Springer Publishing Company, Inc.}, 1st ed., 2010.


\bibitem{kraft}
L.~G.~Kraft, ``A device for quantizing, grouping, and coding amplitude modulated pulses,'' \emph{Cambridge, MA: MS Thesis, Electrical Engineering Department, Massachusetts Institute of Technology,} 1949.

\bibitem{Levenshtein} 
V.I. Levenshtein, 
\newblock {``On perfect codes in deletion and insertion metric,''}
\newblock {\em Discrete Math. Appl.}, vol. 2, no. 3, pp. 241 - 258, 1992.


\bibitem{lubell} D.~Lubell, ``A short proof of Sperner's lemma,'' \emph{J. Combin. Theory}, vol.~1, no.~2, pp. 299--299, 1966.

\bibitem{Mathon}
R. Mathon and T. van Trung,
\newblock {``Directed $t$-packings and directed $t$-Steiner systems,''}
\emph{Designs, Codes and Cryptography}, vol.~18, no.~1-3, pp. 187--198, 1999.

\bibitem{bmz} 
A.~Mazumdar, A.~Barg, and G.~Zemor, ``Constructions of rank modulation codes,''
  in \emph{Proc. IEEE Intern. Symp. on Inform. Theory}, 
	pp. 869 --873, Jul./Aug. 2011.
	
\bibitem{mcmillan}
B.~McMillan, ``Two inequalities implied by unique decipherability,'' \emph{IEEE Trans. on Inform. Theory} vol.~2, no.~4, pp.115--116, 1956.
	
\bibitem{meshalkin} 
L.~D.~Meshalkin, ``Generalization of Sperner's theorem on the number of subsets of a finite set,'' \emph{Theory Probab. Appl.} vol.~8, no.~2, pp. 203--204, 1963.

\bibitem{slepian}
D.~Slepian, ``Permutation modulation,'' \emph{Proc. IEEE,} vol.~53, no.~3, pp. 228--236, Mar. 1965.

\bibitem{ulam}
S. Ulam,
\newblock {``Monte-Carlo calculations in problems of mathematical physics,''}
\newblock {\em Modern Mathematics for the Engineer, Second Series}, (E. Beckenbach, ed.), pp. 261 - 281, 1961.
  
   \bibitem{yamamoto}
K.~Yamamoto, ``Logarithmic order of free distributive lattice,'' \emph{J. Math. Soc. Japan} vol.~6, pp. 343--353, 1954.

\end{thebibliography}
\end{document}